\documentclass[a4paper,11pt,dvipdfmx]{article}

\usepackage{graphicx}
\usepackage{amsmath,amssymb,amsthm}
\usepackage{ascmac}

\usepackage{newpxtext,newpxmath}

\usepackage{url}

\usepackage{a4,latexsym,graphics}

\RequirePackage{color}
\theoremstyle{plain}
\newtheorem{theorem}{Theorem}
\newtheorem{proposition}[theorem]{Proposition}
\newtheorem{lemma}[theorem]{Lemma}

\theoremstyle{definition}
\newtheorem{definition}{Definition}
\newtheorem{example}{Example}
\newtheorem{claim}{Claim}

\usepackage{colortbl}

\newcommand{\ci}{\mathrm{CI}}
\newcommand{\cratio}{\mathrm{CR}}
\newcommand{\randindex}{\mathrm{RI}}
\newcommand{\ciast}{\ci^\ast}
\newcommand{\cthree}{c_3}

\newcommand{\lambdamax}{\lambda_\mathrm{max}}
\newcommand{\submat}[2][A]{{#1^{(#2)}}}
\newcommand{\sumset}{I_3}
\newcommand{\sumindex}[1]{{\sumset^{(#1)}}}

\begin{document}

\title{Size Independence of Consistency Index for Pairwise Comparison Matrices in Analytic Hierarchy Process}
\author{Tsuneshi OBATA\thanks{Faculty of Science and Engineering, Otemon Gakuin University, 2-1-15 Nishiai, Ibaraki, Osaka, 567-8502 Japan \\ {e-mail: t-obata@haruka.otemon.ac.jp} ORCID: 0000-0002-5924-112X} \and 
Shunsuke SHIRAISHI\thanks{Faculty of Applied Information Science, Hiroshima Institute of Technology, 2-1-1 Miyake, Saeki-ku, Hiroshima 731-5193 Japan \\ {e-mail: s.shiraishi.wx@it-hiroshima.ac.jp} ORCID: 0000-0003-2401-1938}}

\date{}

\maketitle

\abstract{
	Pairwise comparisons are fundamental in the analytic hierarchy process.
	Various consistency indices have been proposed to assess inconsistencies in these comparisons.
	Since Saaty first proposed his consistency index, the assessment of the degree of consistency in pairwise comparison matrices has remained an open and hot topic in the study of the analytic hierarchy process.
	The consistency indices $\ci$ and $\cratio$ proposed by Saaty are defined using the principal eigenvalue of the pairwise comparison matrix.
	In our previous study, we introduced an alternative index derived from the relationship between the coefficient of the characteristic polynomial and the consistency of comparisons.
	\newline\indent
	Saaty proposed a fixed threshold of 0.1 for $\ci$ or $\cratio$ as a guideline for an acceptable level of consistency, regardless of the matrix size.
	However, whether this threshold represents an equivalent level of consistency across different matrix sizes, that is, across different numbers of evaluation items, remains unclear.
	This study analysed the relationship between consistency and matrix size by examining pairwise comparison matrices constructed from subsets of evaluation items.
	Based on this analysis, we propose the fundamental property to be satisfied by a size-independent consistency index.
	\newline\indent
	Furthermore, we refine our previously proposed index to ensure that it satisfies this property, demonstrating that it coincides with the existing consistency index.
	Finally, we visualise the relationship between the matrix size and consistency index values using randomly generated pairwise comparison matrices, thereby providing insights into the impact of matrix size on consistency evaluation.
}

\section{Introduction}

Pairwise comparison matrices (PCMs) are fundamental tools in the analytic hierarchy process (AHP) and other decision making methods, which enable decision makers to express their relative preferences between items\footnote{In this paper, we employ the term \emph{items} to refer to the objects being compared in AHP, such as criteria, sub-criteria, and alternatives.}.
As human judgements are rarely perfectly consistent, various \emph{consistency indices} have been proposed to quantify deviations from perfect consistency.
The assessment of the degree of consistency in PCMs has remained an open and hot topic~\cite{brunelli_2018,mazurek_book}.
The consistency index ($\ci$) and consistency ratio ($\cratio$) introduced by Saaty~\cite{saaty_book,saaty_book2} are widely employed in practice.

The interpretation of the consistency values across PCMs of varying sizes remains a long-standing question~\cite{brunelli_book,pelaez-lamata}.
For example, in Saaty's framework, the thresholds $\ci \leq 0.1$ and $\cratio \leq 0.1$ are applied uniformly regardless of the matrix size $n$.
However, the same numerical value for a particular index does not necessarily represent the degree of inconsistency as $n$ changes.
This raises the question of the properties that a \emph{size-independent} consistency index should possess.

In this study, we address the size-independence problem by analysing the relationship between a PCM of size $n$ (referred to as a \emph{super-PCM}) and $n$ distinct PCMs of size $n-1$ obtained by eliminating one item (referred to as \emph{sub-PCMs}).
The consistency of a super-PCM should be comparable with that of its sub-PCMs.
We formalise this intuition by introducing a property that links the index values of a super-PCM and its sub-PCMs, and discuss its implications for interpreting consistency across different matrix sizes.

Building on this property, we refine an existing index and identify the consistency measure that fulfils the property.
We further examine whether several well-known indices satisfy this property, and conduct numerical experiments with randomly generated PCMs to investigate their behaviour.
The results highlight the differences in how existing indices respond to changes in matrix size and suggest directions for defining or selecting indices that are more robust to such variations.

The main contributions of this study are as follows.
\begin{enumerate}
	\item Introducing the concept of sub-PCM and super-PCM, and analysing the relationship between them.
	\item Proposing a property that characterises size-independence for consistency indices.
	\item Identifying a consistency index that satisfies this property.
	\item Using numerical experiments to illustrate and compare the behaviour of representative indices.
	\item Discussing the potential of sub-PCM analysis.
\end{enumerate}

The remainder of this paper is organised as follows.
Section~\ref{sec:pairwise_comparison} reviews the fundamentals of pairwise comparisons and consistency indices, including their definitions and known properties.
Section~\ref{sec:sub_pcm} investigates the relationship between the consistency of a PCM and its size, introduces the concepts of super-PCMs and sub-PCMs, and examines their connections.
Section~\ref{sec:experiments} presents numerical experiments using random PCMs to visualise and compare the behaviours of several indices with respect to the average-preserving property.
Section~\ref{sec:conclusions} concludes the paper and discusses the directions for future research.

\section{Pairwise Comparison and Consistency}
\label{sec:pairwise_comparison}

\subsection{Pairwise Comparison Matrix}

Suppose there are $n$ items to be evaluated, denoted by $C_1, C_2, \dots, C_n$.
The decision maker performs pairwise comparisons for all possible pairs $(C_i, C_j)$, assessing their relative preferences (importance or superiority) by judging which item is preferable.
As there are $\binom{n}{2} = n(n-1)/2$ pairs, $n(n-1)/2$ comparisons are required.
Unless stated otherwise, we assume $n \geq 3$ throughout this paper.

Pairwise comparisons are performed using the verbal expressions listed in Table~\ref{table:saaty_scale}.
Saaty interpreted these expressions as reflecting the ratio of preference between two items in the decision maker's mind, and proposed assigning them numerical values from $1$ to $9$ and their reciprocals.
This numerical representation is referred to as Saaty's scale.
\begin{table}
	\centering
	\caption{Verbal expressions used in Saaty's scale}
	\label{table:saaty_scale}
	\begin{tabular}{lcc}
		\hline
		\multicolumn{1}{c}{Verbal expression}          & Value of $a_{ij}$ & Value of $a_{ji}$ \\
		\hline
		$C_j$ is absolutely more preferable than $C_i$ & $1/9$             & $9$ \\
		$C_j$ is moderately more preferable than $C_i$ & $1/7$             & $7$ \\
		$C_j$ is strongly more preferable than $C_i$   & $1/5$             & $5$ \\
		$C_j$ is weakly more preferable than $C_i$     & $1/3$             & $3$ \\
		$C_i$ is equally preferable to $C_j$           & $1$               & $1$ \\
		$C_i$ is weakly more preferable than $C_j$     & $3$               & $1/3$ \\
		$C_i$ is strongly more preferable than $C_j$   & $5$               & $1/5$ \\
		$C_i$ is moderately more preferable than $C_j$ & $7$               & $1/7$ \\
		$C_i$ is absolutely more preferable than $C_j$ & $9$               & $1/9$ \\
		\hline
		\multicolumn{3}{l}{The intermediate values $2, 4, 6, 8$ and their reciprocals may also be used.}
	\end{tabular}
\end{table}

As a result of the comparison between $C_i$ and $C_j$, a numerical value $a_{ij}$ representing the relative preference of $C_i$ over $C_j$ is assigned based on Table~\ref{table:saaty_scale}.
Its reciprocal $a_{ji} = 1/a_{ij}$, which represents the relative preference of $C_j$ over $C_i$, is simultaneously determined.
When $C_i$ is compared with itself, it is considered equally preferable; therefore, $a_{ii}$ should be set to $1$.
By arranging all these values $a_{ij}$, we obtain an $n \times n$ square matrix $A = (a_{ij})$.

\begin{definition}
	An $n \times n$ square matrix $A = (a_{ij})$ is called a \emph{pairwise comparison matrix} (\emph{PCM}) or a \emph{positive reciprocal matrix} if, for all $i,j \in \{1,2,\dots,n\}$ it satisfies
	\begin{align*}
		 & a_{ij} > 0,
		\\
		 & a_{ij} = \frac{1}{a_{ji}}.
	\end{align*}
\end{definition}

Consequently, all the diagonal elements of the PCM are equal to $1$.
\begin{equation}
	A =
	\begin{pmatrix}
		1        & a_{12}   & \cdots{} & a_{1n}   \\
		1/a_{12} & 1        & \cdots{} & a_{2n}   \\
		\vdots{} & \vdots{} & \ddots{} & \vdots{} \\
		1/a_{1n} & 1/a_{2n} & \cdots{} & 1
	\end{pmatrix}
	.
	\label{eq:pcm_n}
\end{equation}

Note that, in this study, a PCM is not restricted to matrices whose entries are limited to Saaty's scale, that is, $1, 2, \dots, 9$ and their reciprocals.

\subsection{Consistency and Saaty's Consistency Indices}

In AHP, the $(i,j)$ entry $a_{ij}$ of a pairwise comparison matrix is interpreted as representing the preference ratio between $C_i$ and $C_j$.
Let $w_1, w_2, \dots, w_n$ denote the true preference values (hereafter referred to as \emph{weights}) of items $C_1, C_2, \dots, C_n$.
Then, if the decision maker's judgements are perfectly correct, it holds that
\begin{equation}
	a_{ij} = \frac{w_i}{w_j}, \quad \text{for all $i,j \in \{1,2,\dots,n\}$}.
	\label{eq:consistent1}
\end{equation}

\begin{definition}
	PCM $A = (a_{ij})$ is said to be \emph{consistent} if
	\begin{equation}
		a_{ij} \, a_{jk} = a_{ik}, \quad \text{for all $i,j,k \in \{1,2,\dots,n\}$}.
		\label{eq:consistent2}
	\end{equation}
\end{definition}

It can be verified that condition~\eqref{eq:consistent2} is equivalent to the existence of weights $w_1, w_2, \dots, w_n$ such that Equation~\eqref{eq:consistent1} holds.

In the above condition~\eqref{eq:consistent2}, the range of the indices can be restricted.
\begin{proposition}
	PCM $A = (a_{ij})$ is consistent if and only if
	\begin{equation}
		a_{ij} a_{jk} = a_{ik}, \quad \text{for all $i,j,k \in \{1,2,\dots,n\}$ with $i<j<k$}.
		\label{eq:consistent3}
	\end{equation}
\end{proposition}

\begin{proof}
	Consider the case $k < j < i$.
	Then
	\begin{align*}
		a_{ij} a_{jk}
		 & = \frac{1}{a_{ji}} \cdot \frac{1}{a_{kj}}
		= \frac{1}{a_{ki}}
		= a_{ik}.
	\end{align*}
	The remaining cases are presented similarly.
	Hence, if \eqref{eq:consistent3} holds true, then \eqref{eq:consistent2} holds.
\end{proof}

Incidentally, when the values $a_{ij}$ are restricted to discrete levels of Saaty's scale, a completely consistent PCM is rarely obtained.
In particular, when $n = 4$, only 343 of all $17^6 = 24{,}137{,}569$ possible matrices satisfy the condition~\cite{obata_shiraishi_BIC53-3}.
Nevertheless, if the judgements are reasonably accurate, the resulting matrix is expected to be nearly consistent.

Note that pairwise comparisons are independently elicited; hence, global consistency is not ensured.
For three items, $C_i$, $C_j$, and $C_k$, if $C_j$ is judged to be more preferable than $C_k$ and $C_i$ is preferable to $C_j$, then the transitivity would require $C_i$ to be more preferable than $C_k$, as shown in Fig.~\ref{fig:judgements}(a).
In practice, slips of attention, ambiguity, or judgement noise may yield the opposite verdict.
For example, $C_k$ is judged to be more preferable than $C_i$ as shown in Fig.~\ref{fig:judgements}(b), creating a cycle and breaking the consistency.
Such inconsistencies cannot be ruled out a priori, as they are an inherent risk of independent elicitation and motivate the need for quantitative indices to assess inconsistency.

\begin{figure}
	\centering
	\hfill
	\begin{minipage}{0.45\linewidth}
		\centering
		\includegraphics[scale=1]{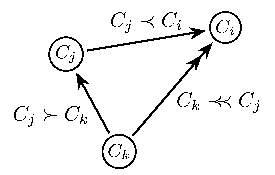}
		\\
		{\small (a) High consistency}
	\end{minipage}
	\hfill
	\begin{minipage}{0.45\linewidth}
		\centering
		\includegraphics[scale=1]{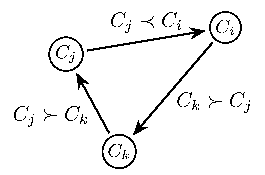}
		\\
		{\small (b) Low consistency}
	\end{minipage}
	\hfill
	\caption{Consistent and inconsistent judgements}
	\label{fig:judgements}
\end{figure}

Therefore, various measures, known as \emph{consistency indices}\footnote{As the number of consistency indices increases as judgements become less consistent, some argue that they are more appropriately termed \emph{inconsistency indices}.}, have been proposed to quantify the degree of consistency by measuring the deviation from a perfectly consistent state~\eqref{eq:consistent2}.

Saaty demonstrated an important property relating the principal eigenvalue $\lambdamax$ of a PCM to its consistency and proposed a consistency index based on this property~\cite{saaty_book}.

\begin{theorem}[Saaty, 1980]
	\label{thm:lambdamax}
	Let $A$ be a PCM of size $n$ and $\lambdamax$ denotes its principal eigenvalue.
	Thus, the following holds:
	\begin{gather*}
		\text{$A$ is consistent} \iff \lambdamax = n;
		\\
		\lambdamax \geq n.
	\end{gather*}
\end{theorem}

\begin{definition}[Saaty's $\boldsymbol{\ci}$]
	For a PCM $A$ of size $n$, the consistency index $\ci$ of $A$ is defined as
	\begin{equation*}
		\ci(A) = \frac{\lambdamax - n}{n - 1}.
	\end{equation*}
\end{definition}

From Theorem~\ref{thm:lambdamax}, the following properties are obtained.
\begin{gather*}
	\text{$A$ is consistent} \iff \ci(A) = 0;
	\\
	\ci(A) \geq 0.
\end{gather*}

As $\ci$ is non-negative and tends to be closer to zero when judgements approach perfect consistency, Saaty regarded it as an index for measuring the degree of consistency in judgements.
Saaty proposed that the larger the value of $\ci$---that is, the farther it lies from zero---the more inconsistent the judgements should be considered.

Hereafter, any function on PCMs that satisfies the same properties is referred to as a consistency index.
\begin{definition}
	A function $c$ defined on the set of PCMs is called a \emph{consistency index} if it satisfies
	\begin{gather*}
		c(A) \geq 0, \quad \text{for any PCM $A$; and}
		\\
		\text{$A$ is consistent} \iff c(A) = 0.
	\end{gather*}
\end{definition}

Although inconsistent judgements cannot be ruled out a priori, Saaty recommended evaluating the value of $\ci$ after completing all pairwise comparisons.
A large value indicates the presence of inconsistent judgements, in which case the pairwise comparisons should be repeated.
Furthermore, he proposed $\ci \leq 0.1$ as a guideline for the acceptable level, although this threshold had no theoretical basis.

In particular, one may question whether the same threshold value of $0.1$ is appropriate irrespective of the matrix size~\cite{brunelli_book,pelaez-lamata}.
Therefore, Saaty generated random PCMs of each size and normalised their $\ci$ values by employing the average $\ci$ value over these random matrices~\cite{saaty_book}.

\begin{definition}[Saaty's $\boldsymbol{\cratio}$]
	For a PCM $A$ of size $n$, the consistency ratio $\cratio$ of $A$ is defined as:
	\begin{equation*}
		\cratio(A) = \frac{\ci(A)}{\randindex(n)},
	\end{equation*}
	where $\randindex$ listed in Table~\ref{table:rand_index}\footnote{%
		In Saaty's books~\cite{saaty_book,saaty_book2}, the values of $\randindex$ for $n = 12$ appear to include typographical errors.
		The values denoted here were obtained from Tone~\cite{tone_book}.} denotes the mean $\ci$-value of the randomly generated PCMs based on Saaty's scale.
\end{definition}
\begin{table}
	\centering
	\caption{Random indices}
	\label{table:rand_index}
	\begin{tabular}{c|ccccccccccccc}
		\hline
		$n$          & $3$    & $4$    & $5$    & $6$    & $7$    & $8$    & $9$    & $10$   & $11$   & $12$   & $13$   & $14$   & $15$ \\
		\hline
		$\randindex$ & $0.58$ & $0.90$ & $1.12$ & $1.24$ & $1.32$ & $1.41$ & $1.45$ & $1.49$ & $1.51$ & $1.53$ & $1.56$ & $1.57$ & $1.59$ \\
		\hline
	\end{tabular}
\end{table}

Note that $\cratio$ also satisfies the properties of a consistency index.

Saaty updated his guideline, replacing the criterion $\ci \leq 0.1$ with $\cratio \leq 0.1$ as the threshold for acceptable consistency.
However, without a clear definition of how to compare the degree of consistency across different matrix sizes, it remains unclear whether this guideline is meaningful.

\subsection{Characteristic Polynomial of a PCM and Consistency}

The consistency of a PCM is closely related to its characteristic polynomial.

Let $P_A(\lambda) = \det(\lambda E - A)$ denote the characteristic polynomial of $n \times n$ PCM~\eqref{eq:pcm_n}:
\begin{equation*}
	P_A(\lambda) = c_0 \lambda^n + c_1 \lambda^{n-1} + c_2 \lambda^{n-2} + c_3 \lambda^{n-3} + \dots + c_n.
\end{equation*}
It can be verified that the coefficients satisfy
\begin{equation*}
	c_0 = 1, \qquad c_1 = -n, \qquad c_2 = 0, \qquad c_n = (-1)^n \det A,
\end{equation*}
and hence
\begin{equation}
	P_A(\lambda) = \lambda^n - n \lambda^{n-1} + c_3 \lambda^{n-3} + \dots + (-1)^n \det A.
	\label{eq:char_poly}
\end{equation}

Furthermore, for coefficient $\cthree$ of degree $n-3$, Shiraishi et al.~\cite{shira_obata_daigo} established the following result:
\begin{theorem}[Shiraishi, Obata, and Daigo, 1998]
	\label{thm:cthree_representation}
	Let $A$ be a PCM of size $n$.
	The coefficient $\cthree$ of the characteristic polynomial~\eqref{eq:char_poly} of $A$ can be expressed as
	\begin{equation}
		\cthree(A) = \sum_{i<j<k} \left(2 - \frac{a_{ij}a_{jk}}{a_{ik}} - \frac{a_{ik}}{a_{ij}a_{jk}}\right).
		\label{eq:cthree_representation}
	\end{equation}
\end{theorem}

The coefficient $\cthree$ exhibits the following properties:
\begin{theorem}[Shiraishi, Obata, and Daigo, 1998]
	\label{thm:cthree}
	Let $A$ be a PCM of size $n$ and $\cthree(A)$ denote the coefficient of degree $n-3$ in the characteristic polynomial of $A$.
	Then, the following properties hold.
	\begin{gather*}
		\text{$A$ is consistent} \iff \cthree(A) = 0;
		\\
		\cthree(A) \leq 0.
	\end{gather*}
\end{theorem}

From Theorem~\ref{thm:cthree}, $-\cthree$ satisfies the properties of a consistency index:
\begin{gather*}
	\text{$A$ is consistent} \iff -\cthree(A) = 0;
	\\
	-\cthree(A) \geq 0.
\end{gather*}
Therefore, $-\cthree$ is proposed as the consistency index~\cite{shira_obata_daigo}.

Note that a consistency index similar to $-\cthree$, denoted by $\ciast$, was proposed by Pel\'{a}ez and Lamata~\cite{pelaez-lamata}.

\begin{definition}[Pel\'{a}ez and Lamata, 2003]
	\begin{equation}
		\ciast(A) = \frac{1}{\binom{n}{3}} \sum_{i<j<k} \left(\frac{a_{ij}a_{jk}}{a_{ik}} + \frac{a_{ik}}{a_{ij}a_{jk}} - 2\right).
		\label{eq:ciast}
	\end{equation}
\end{definition}

Here, $\binom{n}{3}$ denotes the number of distinct combinations of the three items from $n$.
\begin{equation*}
	\binom{n}{3} = \frac{n(n-1)(n-2)}{6},
\end{equation*}
which is equal to the number of terms in the summation on the right-hand side of \eqref{eq:ciast}.

Brunelli et al.~\cite{Burunelli_proportionality} pointed out the relationship between $\cthree$ and $\ciast$.

\begin{proposition}[Brunelli, Critch, and Fedrizzi, 2013, Proposition 1]
	Consider a positive reciprocal matrix $A = (a_{ij})$ with $n \geq 3$, the consistency indices $-\cthree$ and $\ciast$ satisfy the following equality:
	\begin{equation}
		-\cthree(A) = \binom{n}{3} \ciast(A).
		\label{eq:c3_and_ciast}
	\end{equation}
\end{proposition}

Hence, these indices are essentially equivalent and differ only in the scaling factor.

\section{Relationship between Consistency and PCM Size}
\label{sec:sub_pcm}

\subsection{Sub-PCM and Super-PCM}

To investigate how the size of a PCM is related to its consistency, we considered a situation in which one item was removed from the set of items under evaluation.

First, suppose that a decision maker conducts pairwise comparisons on $n$ items, $C_1, C_2, \dots, C_n$.
This yields an $n \times n$ PCM of the form
\begin{equation*}
	A =
	\begin{pmatrix}
		1        & a_{12}   & \cdots{} & a_{1n}   \\
		1/a_{12} & 1        & \cdots{} & a_{2n}   \\
		\vdots{} & \vdots{} & \ddots{} & \vdots{} \\
		1/a_{1n} & 1/a_{2n} & \cdots{} & 1
	\end{pmatrix}.
\end{equation*}

To study the effect of removing one item from a PCM, we define the corresponding submatrix as a sub-PCM, and refer to the original matrix as the corresponding super-PCM.

\begin{definition}
	Let $\submat{s}$ denote the $(n-1) \times (n-1)$ square matrix obtained from a PCM $A$ of size $n$ by eliminating its $s$th row and $s$th column, where $1 \leq s \leq n$.
	We call $\submat{s}$ a \emph{sub-pairwise comparison matrix} (\emph{sub-PCM}) of $A$ with respect to $s$.
	\begin{equation*}
		\submat{s} =
		\begin{pmatrix}
			1           & \cdots{} & a_{1,s-1}     & a_{1,s+1}   & \cdots{} & a_{1n}    \\
			\vdots{}    & \ddots{} & \vdots{}      & \vdots{}    & \ddots{} & \vdots{}  \\
			1/a_{1,s-1} & \cdots{} & 1             & a_{s-1,s+1} & \cdots{} & a_{s-1,n} \\
			1/a_{1,s+1} & \cdots{} & 1/a_{s-1,s+1} & 1           & \cdots{} & a_{s+1,n} \\
			\vdots{}    & \ddots{} & \vdots{}      & \vdots{}    & \ddots{} & \vdots{}  \\
			1/a_{1n}    & \cdots{} & 1/a_{s-1,n}   & 1/a_{s+1,n} & \cdots{} & 1
		\end{pmatrix}.
	\end{equation*}

	Conversely, from the perspective of a sub-PCM, the original matrix $A$ is referred to as a \emph{super-pairwise comparison matrix} (\emph{super-PCM}).
\end{definition}

Matrix $\submat{s}$ defined in this manner is a positive reciprocal matrix and should correspond to the PCM that would result if the same decision maker conducted pairwise comparisons on the remaining $n-1$ items $C_1, \dots, C_{s-1}, C_{s+1}, \dots, C_n$.
This perspective enables us to analyse how the consistency of a super-PCM relates to that of its sub-PCMs, which is essential for understanding the influence of the matrix size on consistency indices.

\begin{example}
	Suppose that the five items $C_1, C_2, C_3, C_4, C_5$ are compared in pairs as shown in Fig.~\ref{fig:omit_comparison}(a) and listed in Table~\ref{tbl:comparison5}.
	Then, the following PCM is obtained:
	\begin{equation*}
		A=
		\begin{pmatrix}
			1   & 3 & 1/3 & 1/3 & 1/5 \\
			1/3 & 1 & 1/7 & 1/5 & 1/7 \\
			3   & 7 & 1   & 1   & 5   \\
			3   & 5 & 1   & 1   & 1/3 \\
			5   & 7 & 1/5 & 3   & 1
		\end{pmatrix}
	\end{equation*}
	\begin{table}
		\centering
		\caption{Pairwise comparison results for five items}
		\label{tbl:comparison5}
		\begin{tabular}{c|cccc}
			\hline
			      & $C_2$       & $C_3$           & $C_4$         & $C_5$ \\
			\hline
			$C_1$ & weakly more & weakly less     & weakly less   & strongly less \\
			$C_2$ & ---         & moderately less & strongly less & moderately less \\
			$C_3$ & ---         & ---             & equal         & strongly more \\
			$C_4$ & ---         & ---             & ---           & weakly less \\
			\hline
		\end{tabular}
	\end{table}

	In this case, if $C_3$ is excluded in the initial set of items as shown in Fig.~\ref{fig:omit_comparison}(b), the PCM obtained by the same decision maker would be
	\begin{equation*}
		\left(
		  \begin{array}{cc>{\columncolor[gray]{0.6}}ccc}
			  1                      & 3 & 1/3 & 1/3 & 1/5 \\
			  1/3                    & 1 & 1/7 & 1/5 & 1/7 \\
			  \rowcolor[gray]{0.6} 3 & 7 & 1   & 1   & 5   \\
			  3                      & 5 & 1   & 1   & 1/3 \\
			  5                      & 7 & 1/5 & 3   & 1
		  \end{array}
		\right)
		\quad\rightarrow\quad
		\begin{pmatrix}
			1   & 3 & 1/3 & 1/5 \\
			1/3 & 1 & 1/5 & 1/7 \\
			3   & 5 & 1   & 1/3 \\
			5   & 7 & 3   & 1
		\end{pmatrix}
	\end{equation*}
	which is precisely the sub-PCM $\submat{3}$ of $A$.

	\begin{figure}
		\centering
		\hfill
		\begin{minipage}{0.45\linewidth}
			\centering
			\includegraphics[scale=1]{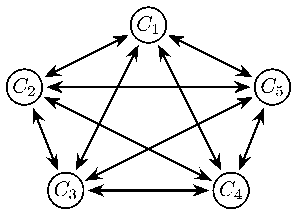}
			\\
			{\small (a) Comparison among five items}
		\end{minipage}
		\hfill
		\begin{minipage}{0.45\linewidth}
			\centering
			\includegraphics[scale=1]{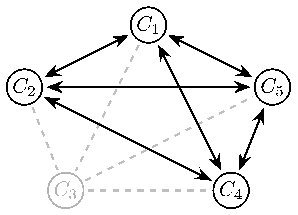}
			\\
			{\small (b) Comparison among remaining items}
		\end{minipage}
		\hfill
		\caption{Omitting one item from pairwise comparisons}
		\label{fig:omit_comparison}
	\end{figure}

\end{example}

\subsection{Desirable Property for Size Independence}

The consistency of a super-PCM is expected to be comparable to that of its sub-PCMs.

\begin{claim}
	The consistency of a super-PCM should be comparable to that of its sub-PCMs.
\end{claim}

Therefore, this study requires a consistency index such that the same value represents comparable degrees of consistency regardless of the matrix size $n$.
In particular, we focus on the property that a super-PCM and its sub-PCMs yield comparable index values, and investigate the indices that satisfy this property.

Thus, the following holds:
\begin{proposition}
	A super-PCM is consistent if and only if all sub-PCMs are consistent.
\end{proposition}
\begin{proof}
	If a super-PCM is consistent, all its sub-PCMs are consistent.

	Conversely, it is assumed that all sub-PCMs of PCM $A$ of size $n$ are consistent.
	From the consistency of $\submat{1}$, we obtain that
	\[
		a_{ij}a_{jk}=a_{ik}, \quad \text{for all $i,j,k \in \{2,3,\dots,n\}$ with $i<j<k$}.
	\]
	From the consistency of $\submat{n}$, we obtain that
	\[
		a_{1j}a_{jk}=a_{1k}, \quad \text{for all $j,k \in \{2,3,\dots,n-1\}$ with $j<k$}.
	\]
	From the consistencies of $\submat{2}$ and $\submat{3}$, we obtain that
	\[
		a_{1j}a_{jn}=a_{1n}, \quad \text{for all $j \in \{3,\dots,n-1\}$}, \qquad \text{and} \qquad
		a_{12}a_{2n}=a_{1n}.
	\]

	Combining these relations yields
	\[
		a_{ij}a_{jk}=a_{ik}, \quad \text{for all $i,j,k \in \{1,2,\dots,n\}$ with $i<j<k$},
	\]
	This proves that $A$ is consistent.
\end{proof}

Hereafter, unless otherwise specified, we assume that a super-PCM is of size $n \geq 4$ and is inconsistent.

Although we claim that the consistency of a super-PCM is comparable to that of its sub-PCMs, inconsistency may emerge from a single pairwise judgement.
The sub-PCM obtained by excluding the items involved in such a comparison is expected to exhibit higher consistency than when those items are retained.
Because the consistency of the sub-PCMs derived from the same super-PCM can vary, it is unrealistic to assume that the degree of consistency of all the sub-PCMs coincides exactly with that of the super-PCM.

Therefore, we revise the earlier claim as follows:
\begin{claim}
	The average consistency of the sub-PCMs is comparable to that of the super-PCM.
\end{claim}

Based on this claim, we propose the following property that a consistency index should satisfy.

\begin{definition}
	A consistency index $c(A)$ is said to have the \emph{average-preserving property} if, for any PCM $A$ of size $n$ $(n \geq 4)$, the following condition holds:
	\begin{equation*}
		c(A) = \frac{1}{n} \sum_{s=1}^{n} c(\submat{s}).
	\end{equation*}
\end{definition}

This property may be regarded as desirable for a consistency index to be size-independent.

\subsection{Consistency Index Satisfying the Average-Preserving Property}

Therefore, we focus on the consistency indices derived from $-\cthree$ that satisfy the average-preserving property.
To this end, we first examine the relationship between the $\cthree$ values of a super-PCM and those of its sub-PCMs.

In the representation of $\cthree$ provided in~\eqref{eq:cthree_representation}, it is convenient to explicitly obtain the range of summation indices $i,j,k$ as follows:
\begin{equation*}
	\sumset = \left\{ (i,j,k) \in \{1,2,\dots,n\}^3 \,\mid \, i < j < k \right\}.
\end{equation*}
For brevity, we introduce the following notation:
\begin{equation*}
	a(i,j,k) = 2 - \frac{a_{ij}a_{jk}}{a_{ik}} - \frac{a_{ik}}{a_{ij}a_{jk}},
	\qquad (i,j,k) \in \sumset,
\end{equation*}
such that $\cthree$ can be expressed concisely as
\begin{equation*}
	\cthree(A) = \sum_{(i,j,k) \in \sumset} a(i,j,k).
\end{equation*}

Similarly, to express the $\cthree$-value of a sub-PCM using the same notation, we obtain
\begin{equation*}
	\cthree(\submat{s}) = \sum_{(i,j,k) \in \sumindex{s}} a(i,j,k),
\end{equation*}
where
\begin{equation*}
	\sumindex{s} = \left\{ (i,j,k) \in \{1,\dots,s-1,s+1,\dots,n\}^3 \,\mid \, i < j < k \right\}.
\end{equation*}

The following relationship holds between the $\cthree$-values of a super-PCM and its sub-PCMs.
The $\cthree$ value of the super-PCM can be considered as being distributed among its sub-PCMs.
The lemma below establishes the precise relationship between these quantities.

\begin{lemma}
	\label{thm:c3_super_sub}
	For a PCM $A$ of size $n$ $(n \geq 4)$, the following equality holds.
	\begin{equation}
		\cthree(\submat{1}) + \cthree(\submat{2}) + \dots + \cthree(\submat{n})
		= (n-3) \cthree(A).
		\label{eq:cmod_sub}
	\end{equation}
\end{lemma}

\begin{proof}
	For a sub-PCM $\submat{s}$, the index range $\sumindex{s}$ is obtained from $\sumset$ by removing all triples $(i,j,k)$, where one of $i,j,k$ coincides with $s$.

	If we fix an element $(s_1,s_2,s_3) \in \sumset$, then this triple does not belong to $\sumindex{s_1}$, $\sumindex{s_2}$, or $\sumindex{s_3}$ but belongs to each of the remaining $n-3$ sets.
	\begin{equation*}
		\sumindex{s}, \quad s \in \{1,2,\dots,n\}\backslash{}\{s_1, s_2, s_3\}.
	\end{equation*}
	Therefore, in the sum
	\begin{equation*}
		\cthree(\submat{1}) + \cthree(\submat{2}) + \dots + \cthree(\submat{n})
		= \sum_{s=1}^n \sum_{(i,j,k) \in \sumindex{s}} a(i,j,k),
	\end{equation*}
	term $a(s_1,s_2,s_3)$ appears exactly $n-3$ times.
	As this is true for every element of $\sumset$, we obtain
	\begin{equation*}
		\cthree(\submat{1}) + \cthree(\submat{2}) + \dots + \cthree(\submat{n})
		= (n-3) \cthree(A).
	\end{equation*}
\end{proof}

For the case $n = 5$, Table~\ref{tbl:index_set} shows whether each element of $\sumset$ is contained in $\sumindex{1}, \dots, \sumindex{5}$.
Each of $\sumindex{1}, \dots, \sumindex{5}$ consists of four elements, and in total (counting duplicate elements separately) there are $5 \times 4 = 20$ elements.
This number is equal to $(5-3) \times 10 = 20$, where $10$ is the number of elements in $\sumset$.
\begin{table}
	\centering
	\caption{Relationship between the elements of $\sumset$ and $\sumindex{s}$ for $n = 5$}
	\label{tbl:index_set}
	\begin{tabular}{c|ccccc}
		\hline
		          & $\sumindex{1}$ & $\sumindex{2}$ & $\sumindex{3}$ & $\sumindex{4}$ & $\sumindex{5}$ \\
		\hline
		$(1,2,3)$ &                &                &                & $\checkmark$   & $\checkmark$ \\
		$(1,2,4)$ &                &                & $\checkmark$   &                & $\checkmark$ \\
		$(1,2,5)$ &                &                & $\checkmark$   & $\checkmark$   & \\
		$(1,3,4)$ &                & $\checkmark$   &                &                & $\checkmark$ \\
		$(1,3,5)$ &                & $\checkmark$   &                & $\checkmark$   & \\
		$(1,4,5)$ &                & $\checkmark$   & $\checkmark$   &                & \\
		$(2,3,4)$ & $\checkmark$   &                &                &                & $\checkmark$ \\
		$(2,3,5)$ & $\checkmark$   &                &                & $\checkmark$   & \\
		$(2,4,5)$ & $\checkmark$   &                & $\checkmark$   &                & \\
		$(3,4,5)$ & $\checkmark$   & $\checkmark$   &                &                & \\
		\hline
	\end{tabular}
\end{table}

It follows from Lemma~\ref{thm:c3_super_sub} that the following holds:
\begin{theorem}
	\label{thm:modifying_c3}
	For $n \geq 3$, the function defined for PCM $A$ of size $n$ as
	\begin{equation*}
		c(A) = -\frac{6}{n(n-1)(n-2)}\cthree(A)
	\end{equation*}
	is a consistency index and satisfies the average-preserving property.
\end{theorem}

\begin{proof}
	As $c(A)$ is clearly a consistency index, it remains to be demonstrated that it satisfies the average-preserving property.

	Define
	\begin{equation*}
		p(n) = \frac{6}{n(n-1)(n-2)}.
	\end{equation*}
	Then, the index can be expressed as
	\begin{equation*}
		c(A) = -p(n)\cthree(A).
	\end{equation*}

	From Equation~\eqref{eq:cmod_sub} of Lemma~\ref{thm:c3_super_sub}, we obtain
	\begin{align*}
		\frac{1}{n} \sum_{s=1}^{n} c(\submat{s})
		 & = -\frac{p(n-1)}{n} \sum_{s=1}^{n} \cthree(\submat{s}) \\
		 & = -\frac{(n-3)\,p(n-1)}{n} \cthree(A)                  \\
		 & = \frac{(n-3)\,p(n-1)}{n\,p(n)}\,c(A).
	\end{align*}
	Finally, substituting $p(n)$ and $p(n-1)$ into this expression yields
	\begin{align*}
		\frac{1}{n} \sum_{s=1}^{n} c(\submat{s})
		 & = \frac{(n-3)\,p(n-1)}{n\,p(n)}\,c(A)                                             \\
		 & = \frac{n-3}{n} \cdot \frac{6}{(n-1)(n-2)(n-3)} \cdot \frac{n(n-1)(n-2)}{6}\,c(A) \\
		 & = c(A).
	\end{align*}
	This confirms the average-preserving property and completes the proof.
\end{proof}

From Equation~\eqref{eq:c3_and_ciast}, the consistency index is \emph{exactly equivalent} to $\ciast$.
This is particularly important because it demonstrates that $\ciast$ is a consistency index that satisfies the desirable property of being \emph{independent} of the matrix size.

\subsection{Saaty's $\boldsymbol{\ci}$ and Average-Preserving Property}
\label{sec:ci_and_property}

While $\ci$ does not share the strong property held by $\ciast$, it still offers the feature that a theoretical upper bound can be established between a super-PCM and its sub-PCMs.

Let $\rho(A)$ denote the principal eigenvalue of a square matrix $A$.
The consistency index $\ci$ proposed by Saaty is expressed as
\begin{equation*}
	\ci(A) = \frac{\rho(A) - n}{n-1}.
\end{equation*}
According to this definition, the relationship between $\ci$ and $n$ is strongly influenced by the relationship between the principal eigenvalue $\rho(A)$ and $n$.

Horn and Johnson~\cite{horn_book} provided a detailed account of the relationship between the eigenvalues of a matrix and its principal submatrices.
To state their results, we first recall the following definitions:
\begin{definition}[Horn and Johnson, 2012]
	Let $A$ be an $n \times n$ square matrix and $\alpha \subset \{1,2,\dots,n\}$.
	The matrix obtained by deleting all rows and columns whose indices belong to $\alpha$ from $A$ is denoted by $A[\alpha]$ and is called the \emph{principal submatrix} of $A$.
\end{definition}

Any sub-PCM is the principal submatrix of a super-PCM.

In particular, the following result by Horn and Johnson provides useful properties of the principal eigenvalue in relation to the principal submatrices:

\begin{theorem}[Horn and Johnson, 2012, Corollary 8.1.20]
	\label{thm:horn-johnson}
	Let $A=(a_{ij})$ be a non-negative square matrix.
	Thus, the following statements hold:
	\begin{enumerate}
		\item \label{item:submatrix}
		      If $\tilde{A}$ is a principal submatrix of $A$, then $\rho(\tilde{A}) \leq \rho(A)$.
		\item $\max_{i=1,2,\dots,n} a_{ii} \leq \rho(A)$.
		\item If all diagonal entries of $A$ are positive, then $\rho(A) > 0$.
	\end{enumerate}
\end{theorem}

Let $A$ be a PCM of size $n$.
As all its entries are positive, $A$ satisfies the assumptions of Horn and Johnson.
Accordingly, we obtain the following bounds relating the consistency index of a super-PCM to that of its sub-PCMs:

\begin{theorem}
	Let $A$ be a PCM of size $n$ $(n \geq 4)$ and $\submat{s}$ be a sub-PCM for some $s \in \{1,2,\dots,n\}$.
	Then, the following inequalities hold.
	\begin{align*}
		 & \ci(\submat{s}) \leq \frac{(n-1)\,\ci(A) + 1}{n-2},
		\\
		 & \ci(A) \geq \frac{(n-2)\,\ci(\submat{s})}{n-1} - 1.
	\end{align*}
\end{theorem}

\begin{proof}
	From statement~\ref{item:submatrix} of Theorem~\ref{thm:horn-johnson}, we obtain
	\begin{equation}
		\rho(A) \geq \rho(\submat{s}), \quad \text{for all $s \in \{1,2,\dots,n\}.$}
		\label{eq:eigen_rel}
	\end{equation}

	According to the definition of $\ci$,
	\begin{align*}
		 & \rho(A) = (n-1)\ci(A) + n,
		\\
		 & \rho(\submat{s}) = (n-2)\ci(\submat{s}) + (n - 1).
	\end{align*}
	Substituting these into \eqref{eq:eigen_rel} yields
	\begin{equation*}
		(n-1)\ci(A) + n \geq (n-2)\ci(\submat{s}) + (n - 1).
	\end{equation*}
	Rearranging yields
	\begin{equation*}
		\ci(\submat{s}) \leq \frac{(n-1)\ci(A)+1}{n-2}.
	\end{equation*}
	This proves the first point of the inequality.

	The second inequality in the theorem follows by solving the same relationship for $\ci(A)$.
\end{proof}

\section{Numerical Experiments Using Random PCMs}
\label{sec:experiments}

To investigate how the index values of the super-PCMs relate to those of their sub-PCMs, we conducted numerical experiments with randomly generated PCMs under the following settings.

\begin{itemize}
	\item \textbf{Matrix sizes:} $n = 4, 5, 6, 7$.
	\item \textbf{Number of samples:} For each size $n$, we generated 3,000 random PCMs based on the following procedure:
	      \begin{enumerate}
		      \item Generate $\binom{n}{2} = n(n-1)/2$ random numbers, each uniformly drawn from Saaty's 17-point scale ($1/9, 1/8, \dots, 1/2, 1, 2, \dots, 8, 9$).
		      \item Assign these numbers sequentially to the upper triangular entries of the PCM.
		      \item Fill each lower triangular entry with the reciprocal of its symmetric upper entry and set all diagonal entries to $1$.
	      \end{enumerate}
	\item \textbf{Indices computed:} For each generated PCM of size $n$ (super-PCM) and its $n$ sub-PCMs, we computed three consistency indices: $\ciast = -\frac{1}{\binom{n}{3}}\cthree$, $\ci$, and $\cratio$.
\end{itemize}

The results are shown in Fig.~\ref{fig:super_vs_sub}.
The horizontal axis shows the index values of the super-PCMs, and the vertical axis shows those of the sub-PCMs.
The grey dots represent the index values of the super-PCMs (horizontal) versus those of the sub-PCMs (vertical).
The black dots represent the index values of the super-PCMs (horizontal) versus \emph{average} index values of the sub-PCMs (vertical).
For reference, a line with slope $1$ passing through the origin is drawn.
If the average-preserving property is satisfied, the black dots are expected to lie along this line.

\begin{figure}
	\centering
	\begin{tabular}{cccc}
		 & $\ciast$ & $\ci$ & $\cratio$ \\
		\raisebox{10mm}{\rotatebox{90}{$n=4$}}
		&
		\includegraphics[width=0.3\linewidth]{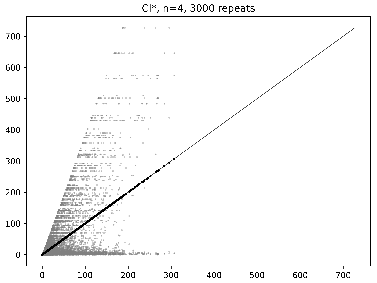}
		&
		\includegraphics[width=0.3\linewidth]{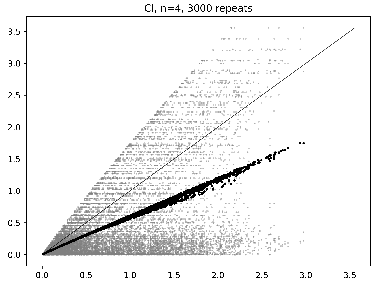}
		&
		\includegraphics[width=0.3\linewidth]{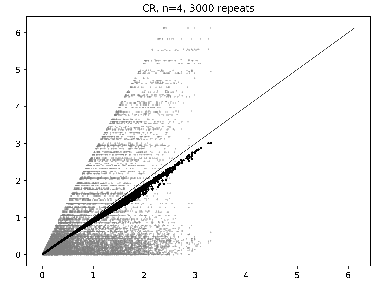}
		\\
		\raisebox{10mm}{\rotatebox{90}{$n=5$}}
		&
		\includegraphics[width=0.3\linewidth]{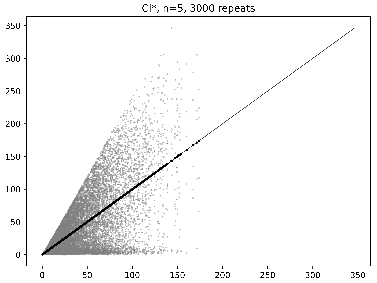}
		&
		\includegraphics[width=0.3\linewidth]{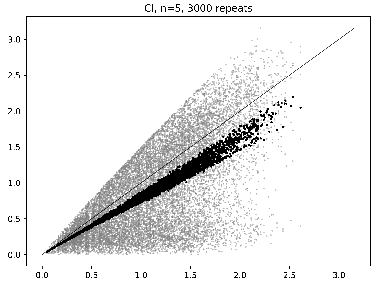}
		&
		\includegraphics[width=0.3\linewidth]{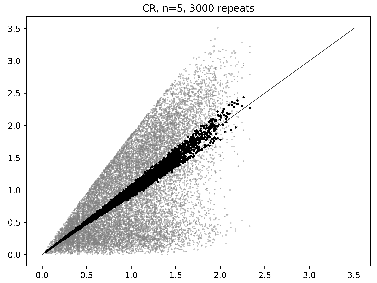}
		\\
		\raisebox{10mm}{\rotatebox{90}{$n=6$}}
		&
		\includegraphics[width=0.3\linewidth]{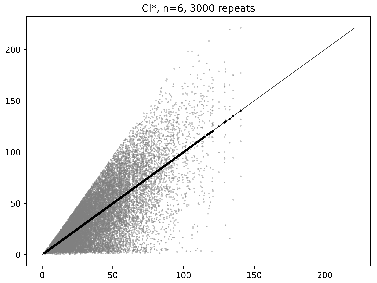}
		&
		\includegraphics[width=0.3\linewidth]{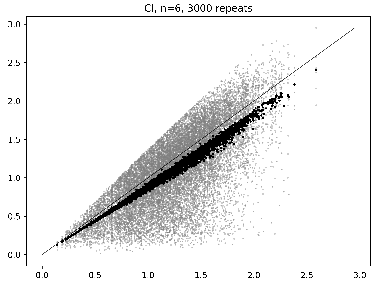}
		&
		\includegraphics[width=0.3\linewidth]{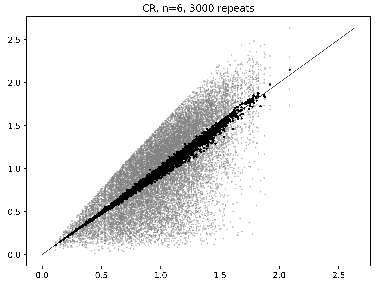}
		\\
		\raisebox{10mm}{\rotatebox{90}{$n=7$}}
		&
		\includegraphics[width=0.3\linewidth]{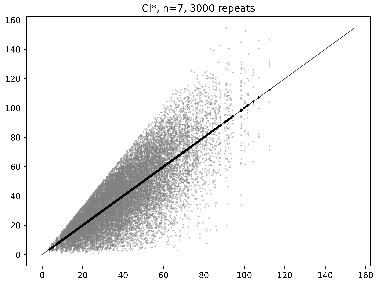}
		&
		\includegraphics[width=0.3\linewidth]{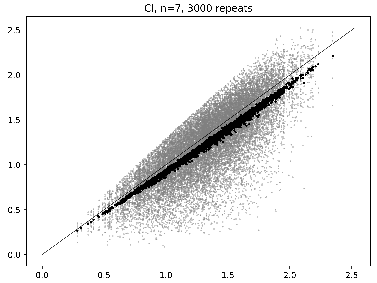}
		&
		\includegraphics[width=0.3\linewidth]{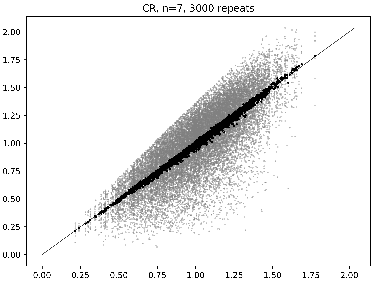}
	\end{tabular}
	\caption{Super-PCM index values (horizontal) vs. sub-PCM index values (vertical).  Grey: individual sub-PCMs; black: averages.}
	\label{fig:super_vs_sub}
\end{figure}

Naturally, for $\ciast$ satisfying the average-preserving property, the black dots lie precisely on a line with a slope $1$.
For $\ci$, the black dots deviate considerably from the line when $n=4$ but approach it as $n$ increases.
For $\cratio$, the black dots are off the line when $n=4$, but for $n \geq 5$, they align more closely with it, indicating behaviour consistent with the average-preserving property.
This suggests that, except in the case $n=4$, normalisation by the $\randindex$ has achieved its intended effect.

Another observation is that the vertical spread of the grey dots is larger for $\ciast$ than for $\ci$ and $\cratio$, possibly reflecting a greater variability in the sub-PCM index values derived from the same super-PCM.

In addition, for all indices, the upper bound of the sub-PCM index values appears to follow a linear pattern, suggesting that they may be bounded by a linear function of the super-PCM index value.
The upper bound for $\ci$ has already been discussed in Section~\ref{sec:ci_and_property}.

Fig.~\ref{fig:upper_bound} shows the scatter plots of $\ci$ with the theoretical upper bound
\begin{equation}
	\ci(\submat{s}) \leq \frac{(n-1)\ci(A)+1}{n-2}
	\label{eq:ci_upper_bound}
\end{equation}
superimposed as dashed lines in the figures.
It can be observed that the $\ci$ values of the sub-PCMs do not reach this limit but remain below it.
Therefore, the inequality~\eqref{eq:ci_upper_bound} is not tight, and the actual strict upper bound may lie slightly lower, suggesting the possibility of a sharper bound.

\begin{figure}
	\centering
	\begin{tabular}{cc}
		\includegraphics[width=0.4\linewidth]{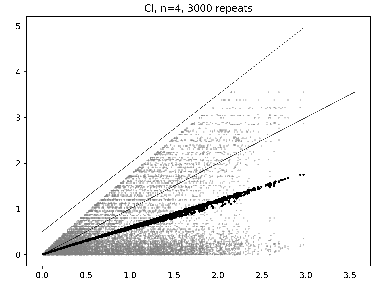}
		      &
		\includegraphics[width=0.4\linewidth]{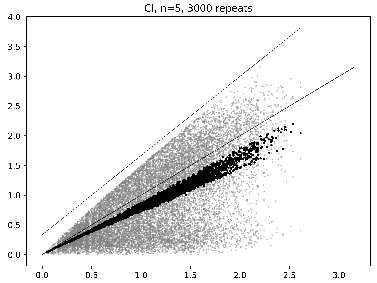}
		\\[-0.5ex]
		$n=4$ & $n=5$
		\\[2ex]
		\includegraphics[width=0.4\linewidth]{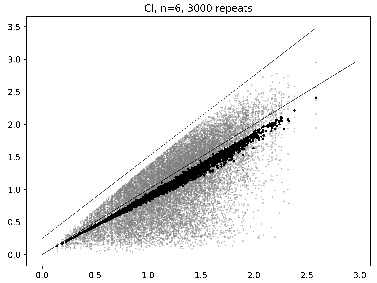}
		      &
		\includegraphics[width=0.4\linewidth]{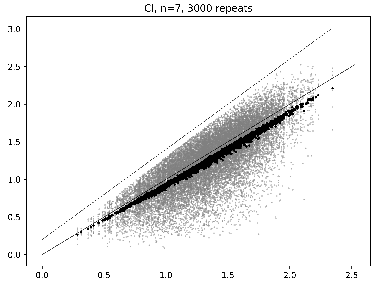}
		\\[-0.5ex]
		$n=6$ & $n=7$
	\end{tabular}
	\caption{Scatter plots of $\ci$ values for super-PCMs and their sub-PCMs, with the theoretical upper bound shown as a line}
	\label{fig:upper_bound}
\end{figure}

\section{Conclusions}
\label{sec:conclusions}

This study focuses on consistency indices that are independent of PCM size.
To clarify the relationship between PCMs of different sizes, we first define the concepts of sub-PCMs and super-PCMs.
Based on the assertion that a sub-PCM and its super-PCM should exhibit comparable consistency, we propose the average-preserving property as a desirable feature of size-independent indices.
We demonstrate that an index equivalent to $\ciast$ satisfies this property, providing a rigorous justification for $\ciast$ as a size-independent consistency index.
Finally, through numerical experiments with randomly generated PCMs, we visualise how the index values of super-PCMs and their sub-PCMs are related and confirm the theoretical results.

The relationship between a super-PCM and its sub-PCMs can be used to identify the comparisons that impair consistency.
Although we assumed that a super-PCM and its sub-PCMs exhibit comparable consistencies, certain comparisons may be inconsistent in practice.
In such cases, the consistency of the sub-PCM obtained by removing items involved in inconsistent comparisons is likely to be higher than that of the super-PCM.
When a super-PCM exhibits low consistency, sub-PCMs with higher consistency can be identified, indicating that the removed items contribute to the overall inconsistency.

The concept of sub-PCMs provides a richer mathematical perspective for analysing PCMs.
Even when two super-PCMs have the same overall index value, the distribution of consistencies of their sub-PCMs can differ significantly.
In some cases, most sub-PCMs are nearly consistent, and only a few extremely inconsistent sub-PCMs lower the overall consistency.
In other cases, all the sub-PCMs show a moderate and uniform level of inconsistency.
Thus, although super-PCMs may appear equally consistent, the \emph{robustness} of their pairwise comparisons differs.
From a mathematical perspective, this variability can be interpreted as revealing a part of the internal structure of PCMs, suggesting avenues for further study in terms of algebraic identities, probabilistic models of random PCMs, and the geometric properties of the PCM space.

Further investigation of these ideas, including the average-preserving property for consistency indices beyond $\ciast$, $\ci$, and $\cratio$, remains a topic for future research.
Developing a more systematic mathematical framework to describe the interaction between super- and sub-PCMs may provide deeper insights into the theoretical foundations of consistency in AHP.

The measurement of consistency in pairwise comparisons is likely to remain an open and hot topic in the future as well.

\subsection*{Acknowledgement}

The authors are grateful to Professor Konrad Ku\l{}akowski for his kind support during the submission process.
This work was supported by JSPS KAKENHI Grant Number 24K07950.
We would like to thank Editage (www.editage.jp) for English language editing.

\end{document}